\documentclass[11pt]{article}
\usepackage{times}
\usepackage{amsthm, amssymb, amsmath}
\usepackage{graphicx}
\usepackage{luca}
\usepackage{hyperref}
\usepackage{tikz}
\usetikzlibrary{arrows,shapes,patterns,fit,calc}

\newtheorem{theorem}{Theorem}

\begin{document}

\title{TrueReview:\\ A Proposal for Post-Publication Peer Review \\ 
{\small \sc white paper}}
\author{Luca de Alfaro \\ Computer Science department \\ University of California, Santa Cruz \\ luca@ucsc.edu \and 
    Marco Faella \\ Electrical Engineering and Information Technologies \\ University of Naples ``Federico II'', Italy \\ m.faella@unina.it}
\date{Technical report UCSC-SOE-16-13 \\ \today}

\maketitle

\sloppy

\def\qcur{q^\mathit{future}}
\def\qprev{q^\mathit{past}}
\def\qtrue{q^\mathit{true}}
\def\goto{\rightarrow}
\def\avg{{\mbox{{\rm avg}}}}
\def\truthf{\theta}
\def\inform{\delta}
\def\bonus{B}
\def\flin{f^\mathit{lin}}
\def\fsigm{f^\mathit{S}}
\def\E#1{{\mbox{E}}\left[ #1 \right]}
\def\avg{{\mbox{{\rm avg}}}}
\def\norm{\mathcal{N}}

\begin{abstract}
In {\em post-publication peer review,} scientific contributions are first published in open-access forums, such as arXiv or other digital libraries, and are subsequently reviewed and possibly ranked and/or evaluated.  Compared to the classical process of scientific publishing, in which review precedes publication, post-publication peer review leads to faster dissemination of ideas, and publicly-available reviews. The chief concern in post-publication reviewing consists in eliciting high-quality, insightful reviews from participants. 

We describe the mathematical foundations and structure of TrueReview, an open-source tool we propose to build in support of post-publication review.
In TrueReview, the motivation to review is provided via an incentive system that promotes reviews and evaluations that are both {\em truthful\/} (they turn out to be correct in the long run) and  {\em informative\/} (they provide significant new information).  TrueReview organizes papers in {\em venues,} allowing different scientific communities to set their own submission and review policies. These venues can be manually set-up, or they can correspond to categories in well-known repositories such as arXiv.  The review incentives can be used to form a \emph{reviewer ranking} that can be prominently displayed alongside papers in the various disciplines, thus offering a concrete benefit to reviewers. The paper evaluations, in turn, reward the authors of the most significant papers, both via an explicit paper ranking, and via increased visibility in search.
\end{abstract}

\section{Introduction}

Peer review has not always preceded publication. 
In the times of Galileo, Newton, van Leeuwenhoek, up to Darwin, scientists would share their results via letters or presentations to scientific societies; the results were then discussed among scientists. 
The current system of pre-publication peer review was widely adopted only relatively recently, starting in the 1940s with the introduction of large-circulation scientific journals \cite{spier_history_2002}.
Pre-publication peer review was shaped by the economics of paper journal publishing: as paper journals are slow and expensive to print and ship, peer-review was used to select which articles deserved wide dissemination.

The economics of publishing is very different now.
Information nowadays can be disseminated immediately at very low 
cost, and furthermore, in a manner that makes it open to social
interaction: in blogs, wikis, forums, social networks, and other
venues, people can both share and comment on information. 
Yet, for the most part the scientific community still beholds pre-publication peer review as the officially anointed method of disseminating results.
Publication in venues such as journals and conferences with a pre-publication peer-review selection process is also the most commonly used measure of scientific productivity, and contributes to shape the near totality of academic and research careers. 

Pre-publication peer review has several drawbacks. 
One of the most salient is the delay imposed on the dissemination of
results. 
In a typical computer science conference, six months may elapse from
submission to publication in the proceedings, and this assumes that
the conference deadline came just when the paper was ready for
submission, and more importantly, that the paper was accepted. 
To avoid this delay, many authors submit the paper to open
repositories such as arXiv\footnote{http://arxiv.org/} at the same
time as they submit it to a conference or journal. 
While this makes it available to other researchers, a submission to
ArXiv does not come with the all-important blessing of peer review. 
As such, papers submitted to arXiv are not generally counted as part of the productivity record of researchers. 
Submitting to arXiv is not an alternative for submissions to
conferences or journals with double-blind review policies, and citing
works submitted in arXiv, but not yet peer-reviewed, is not universally perceived as appropriate in science.

A related issue is the one of selection.  
The current process of scientific reviewing, consequently, aims at
deciding which papers to accept for publication, and which to reject.  
Correctness is only one factor in such a decision: commonly, there are
many more correct submissions (in the sense of exempt from scientific
errors) than can be accepted, and the decision to accept or reject is
motivated by judgements on the significance of the submissions.  
The paper acceptance process is thus of necessity an uncertain
process, where a demarcation line needs to be drawn among papers of
fairly similar apparent significance.  
Papers that present correct results, but which do not make the cut,
are subjected to a delay as they are re-submitted to different
journals of conferences.  
The process is slow and wasteful of resources.  

One last drawback of pre-publication peer review is that papers are not
presented to readers in the context of the accumulated knowledge and
judgement.  
While this shields papers from being presented alongside potentially irrelevant reviews, this also means that insightful observations from readers and researchers cannot help understand papers and put them in context.  

This white paper presents the design principles and mathematical foundations of TrueReview, an open-source system we propose to build in support of post-publication peer-review. 
In the next section, we describe the overall motivations and design principles that inspire the development of TrueReview. 
After a review of related work, we discuss the problem of motivating reviewers, and we describe in Section~\ref{sec:incentives} the reviewer incentive system at the heart of TrueReview. 
The main challenge in a review system consists in ensuring that all papers receive sufficient and precise evaluations.
Our novel incentive scheme promotes reviews that are both truthful and {\em informative,} in the sense that they bring novel information into the system rather than merely confirming what is already known.
To validate the proposed incentive scheme, we report in Section~\ref{sec:sim} the result of simulations of the review process with participants having a varying distribution of skills and paper topic expertise.
The simulations show that the incentive system is effective in ensuring that all papers receive precise evaluations.
We conclude with an overview of the software architecture of TrueReview, and a discussion of some key implementation decisions. 

All the code for TrueReview is open source, and it can be found at \url{https://github.com/TrueReview/TrueReview}.

\section{TrueReview Design Principles}

TrueReview will be an open, on-line system, where authors can publish their papers or enter links to their papers, and where reviewers can review and evaluate the papers after they have been published. 
This would serve the scientific community as a whole, by making the
dissemination of results more open, predictable and less subject to
delays, and by helping researchers view papers in light of the
accumulated knowledge and wisdom.  
At the core of TrueReview are the following design principles.

\paragraph{Driven by scientific communities.}
TrueReview does not plan to be a one-size-fits-all solution for post-publication review. 
Papers published or linked in TrueReview will not be all put into the same ``pot'' for review and ranking. 
Each scientific community has norms for the format of published papers, and has well-known researchers that act as standard-bearers for the community: these are the people that today serve on the journal editorial boards and conference program committees. 
Each community that elects to use TrueReview will decide whether papers are to be submitted to TrueReview directly, or whether TrueReview tracks paper submitted to certain categories in open-access repositories such as Arxiv.
The choice of who can review papers in a venue will also be left to each community.
In some venues, the senior members might wish to approve who has review priviledge, or adopt an invitation system.
For other venues, such as venues that correspond to Arxiv categories, it might be sufficient to have published a paper in the same venue to be able to review.

\paragraph{No delays to publication.}
While papers that have just been submitted are unreviewed, this should
not prevent their circulation. 
One natural objection is whether making papers available immediately
deprives readers from the quality guarantee conferred by a formal
process of paper review. 
We believe that the benefits of the prompt communication of scientific
results far outweigh the drawback of circulating papers in various
stages of review. 
The status of a peer reviewed paper is often assumed by people not
familiar with the process to be a seal of approval that guarantees the
correctness of the results. 
In reality, errors in scientific papers are not always discovered by
the conference or journal review committees to which the papers are
submitted: more often, the errors are discovered by the authors
themselves, or by people who try to use or extend the papers results. 
Only papers that are widely read, and whose results are extensively
used, can be trusted to be highly likely to be correct. 

\paragraph{Rank rather than select.}
When a paper is submitted to a journal or conference, the question
of whether to accept or reject it most often revolves on the
relevance of the paper, rather than on its correctness. 
After the papers that are clearly flawed are eliminated, there are
invariably too many papers to fit in the conference or journal format;
the committee must then select the papers to accept on the basis of
their quality. 
The committee thus essentially performs a ranking task, applying then
a binary threshold dictated by conference or journal constraints. 
For the rejected papers that were indeed correct, this process results
in a pointless delay to publication; as these are typically
resubmitted to other venues, the work that went into ranking them is also wasted.

This summary of the current review process is greatly simplified. 
In truth, there are many conferences and journals, with different
typical quality levels, and authors choose the venue where to submit
the paper in order to compromise between the prestige of the venue,
and the probability that the paper is accepted. 
Nevertheless, the process is wasteful of time and work.
We believe it would be better to use the reviews and comments for
ranking, rather than for selection. 
There would be no need to artificially set a cut-off line; all papers
would be ranked and available on-line as soon as the authors publish them.

Once a ranking of the papers were available (even if approximate),
journals and conferences could use the ranking for selection
purposes. 
For example, a conference could gather people interested in a
particular field, and allocate paper presentation slots to the 30
highest-ranked papers of the year, and poster presentation space to
the next 50 highest-ranked; a journal or book editor could similarly
publish (and distribute to libraries in archival form) the best 50
papers of each year. 
Certainly many users of the system could use the ranking for selection
purposes, but the main goal of the system would be to generate a
ranking, not a selection.

\paragraph{Truthful and informative incentives to review.}
The main obstacle to post-publication review consists in enlisting expert reviewers, and having them provide accurate ratings and reviews on most papers in a venue.
In conference and journals, the enticement to review is provided by the prestige of appearing on the program committee or editorial board of a well known journal or conference. 
Reviewers need to provide competent reviews, or risk appearing uninformed when their reviews are compared with those of others on the program committee or editorial board. 
In exchange for being listed as members of the program committee or editorial board, the reviewers also accept to read and review papers that they would not have read out of their interest alone. 

We plan to recreate the incentive to review by also giving wide publicity to the most active reviewers, and by attributing merit for the reviews via an incentive system that prizes reviews that both are correct, and provide new information.
In each venue, the names of the reviewers that accrued highest review merit will be displayed in the first page, alongside the top-rated papers. 
Reviewers will be able to link their name to a web page of their choice, such as their academic home page.
We hope this visibility and mark of distinction, which mirrors the one currently offered by membership in program committees and editorial boards, will provide sufficient motivation 
to actively participate in the system. 

The incentive system for reviews will reward reviewers who provide rating that are both {\em truthful\/} and {\em informative.}
{\em Truthful\/} ratings are those that will be confirmed later on by the consensus opinion on a paper. 
{\em Informative\/} ratings are those that provide genuinely new information. 
Examples of informative ratings are the first rating for a previously unreviewed paper, or a rating that differs from the current consensus for a paper, but will be later confirmed to be correct. 
In contrast, a rating and review that reflects the consensus on a paper that has already been reviewed many times will have low informative value.
Considering both truthfulness and informativeness to reward reviewers encourages them to focus where their expertise allows them to give new useful information.
This combined incentive should also lead to prompt rating of papers published in venues. 

Additionally, for venues with a long life-span (such as venues replicating arXiv categories),
users can be encouraged to periodically contribute new reviews by slowly decreasing their
accrued score, as long as they do not provide a new review.

\section{Related Work}

A detailed proposal for a post-publication peer-review model was made by Kriegeskorte \cite{openeval12}. 
The author advocates signed reviews and multi-dimensional
paper evaluations, that can be aggregated by different interested parties in different ways.
The dissemination of signed reviews is deemed a sufficient incentive for reviewers to participate in the process.
The paper contains also an in-depth analysis of the benefits of post-publication peer review, which are presented in an eloquent way and which are indeed part of the motivation for this study on incentives.
The incentive schemes we propose do not require review authors to be publicly visible.
This may be beneficial, as there is some evidence that signed reviews may deter prospective reviewers, or dampen the frankness of their opinions \cite{van_rooyen_effect_1999,van_rooyen_effect_2010}.
The virtues, and drawbacks, of signed reviews have been described in \cite{groves_is_2010,khan_is_2010}; signed reviews can prevent the abuse of review power, but they also can stifle criticism.

The proposal for TrueReview shares its fundamental motivations with \cite{das_sarma_crowdsourcing_2011}, of which it represents an evolution, as well as with \cite{openeval12}, while differing in the details of the incentive system.
The incentives we propose do not rule out publishing the names of review authors. 
We propose instead listing, for each publication venue, both the top papers, and the top reviewers as determined by the total of their review bonuses, allowing users of a Web interface to search both papers, and reviewers. 
Post-publication peer review has also been advocated, on similar grounds as \cite{das_sarma_crowdsourcing_2011,openeval12}, in \cite{hunter_post-publication_2012,herron_is_2012,da_silva_need_2013}.
Even the popular press has engaged in the discussion \cite{wired-incentives14},
with the CEO of Academia.edu\footnote{\tt http://www.academia.edu} mentioning the possibility of gathering reputation points
from reviews.

The idea of evaluating scientific proposals via crowdsourcing reviews and ratings has been proposed as a method for adjudicating telescope time, a central issue in Astronomy \cite{merrifield_telescope_2009}, as well as in the evaluation of some National Science Foundation proposals \cite{national_science_foundation_dear_2013}.

ArXiv overlay journals are gaining momentum in several scientific disciplines,
including math, physics, and computer science~\cite{overlays16}.
While their papers are publicly available even before acceptance,
their selection process follows the traditional peer review model of printed journals.
In the words of Timothy Gowers, Fields medalist and managing editor of the
arXiv overlay journal Discrete Analysis, ``our journal is very conventional [...]
But if the model becomes widespread, then I personally would very much like to see more-radical ideas tried out as well''~\cite{gowers13}.

Other organizations are indeed pursuing more radical ideas:
ScienceOpen\footnote{\tt http://www.scienceopen.com} publishes articles online under
an open-access model and encourages post-publication peer reviews, which
include a numerical score.
Reviews are publicly attributed to their authors and even assigned a DOI.
On the other hand, reviewers do not accrue a numerical reputation for their efforts.
Similarly to~\cite{openeval12}, the incentive for the reviewers consists in having a public collection
of their reviews.

O'Peer\footnote{\tt http://opeer.org} is a proof-of-concept website where authors-reviewers accrue 
reputation (called \emph{credibility})
according to both their publication record and the quality of their reviews.

An incentive system that shares many of the design goals with the one we propose for TrueReview has been proposed by Bhattacharjee and Goel \cite{bhattacharjee_algorithms_2007} in their work on incentives for robust ranking in online search.
In the \cite{bhattacharjee_algorithms_2007} proposal, users can place tokens on items in order to place wagers on the quality of the items, much as people can bet on horses at races.
If the ratio between the qualities of two items is different from the ratio between the token amounts, an {\em arbitration opportunity} arises, and a user can move a token from the over-rated item to the under-rated one and gain reputation (an operation that is roughly equivalent to betting a negative dollar on a horse and a positive dollar on another, if negative bets were allowed).
The incentive scheme is truthful, as the incentive is to bring token counts in direct proportionality with qualities, and it also promotes informativeness, as the biggest arbitration opportunities occur for the papers that are most under-valued. 
We made various attempts at adapting \cite{bhattacharjee_algorithms_2007} for post-publication review, before finally opting for the grade-based scheme we propose for TrueReview. 
The main problem we encountered is the  slow start in properly ranking new papers. 
When a paper is added, initially it has no tokens. 
If users can place or move one token at a time, a good paper will require many reviews to receive a proper ranking; if users can move many tokens at once, the vandalism of a single user can cause considerable damage.
Another issue was that the truthfulness and innovativeness incentives are tied together by the arbitration opportunity, and their strengths cannot be independently tuned. 
Ultimately, we felt that the approach proposed in this paper was more flexible and allowed us to better control vandalism. 
We can independently tune the truthfulness and informativeness incentives, and we can adopt a number of aggregation strategies for reviews.

Another related line of inquiry focuses on eliciting honest feedback in the absence of ground truth.
Miller et al.\ \cite{miller05} show that a simple scheme 
based on proper scoring rules~\cite{cooke91} induces reviewers to provide their honest opinion
(similarly to our Theorem~\ref{thm:nash1}).
Along the same line, Jurca and Faltings~\cite{jurca09} study scoring systems that are resistant
to collusions, whereas Dasgupta and Ghosh~\cite{dasgupta13} consider the scenario in which reviewers
can strategically calibrate the amount of effort spent for a review.
None of these works model the \emph{selection} of items by reviewers, which is instead
one of our main concerns.

The idea of validating assertions by considering them wagers on future value, and rewarding thus their accuracy, is the principle at the basis of {\em prediction markets\/} \cite{wolfers_prediction_2004,tziralis_prediction_2012}. 
The arbitration opportunities in prediction markets are in fact conceptually similar to those in \cite{bhattacharjee_algorithms_2007}, except that by having real money involved, the possibility for vandalism is virtually eliminated. 
Indeed, the stock market offers a model for crowdsourcing valuations that both is truthful, and that offers a prize for informativeness.
However, the full working of the market (including the put and call options that are important in betting on future valuations) are vastly more complex than the simple mechanism we presented in this paper, and arguably over-complicated for the task at hand.

There has been much work on peer evaluation, in classroom settings \cite{gehringer_strategies_2000,gehringer_electronic_2001,robinson_calibrated_2001,sadauskas_critviz:_2013} and in MOOCs \cite{piech_tuned_2013}. 
In a classroom or MOOC setting, however, the focus is on obtaining precise and fair evaluations, rather than on incentives to select the items (papers, or homework submissions) to review. 
This because in educational settings, students are usually compelled to perform the peer reviews and evaluations as part of their class work.
Furthermore, as homework submissions share all the same topic, the review assignment can be (and usually, is) performed automatically, again obviating the need for an informative incentive system.

\section{The TrueReview Incentive System for Reviewers} 
\label{sec:incentives}

The crucial challenge for post-publication review consists in ensuring that papers receive adequate reviews and precise evaluations.
There are many mechanisms for ensuring that the set of potential reviewers is capable of writing useful reviews: they can be invited to review, or the privilege of reviewing can be granted automatically to people who have successfully published previously in the same venues. 

The basic user action in TrueReview consists in a user choosing a paper, and providing both a written review, and a numerical rating for the paper. The ratings are then aggregated in a single rating for the whole paper.
TrueReview rewards the author of a review with a review ``bonus''.
In each publication venue, reviewers will be listed according to the total of the bonuses they received: we hope this visibility will provide incentive to review.

The incentive scheme used for assigning the review bonuses should be truthful: the strategy for users to maximize their bonuses for each review should be to express their honest opinion about the paper. 
Furthermore, the incentive scheme should be {\em informative:\/} it should prize new relevant information over repetition of already-known information. 
For instance, it should value the first review on a paper more than a review confirming the consensus opinion on a paper that has already been reviewed many times.
Among papers having the same number of reviews, an informative incentive scheme should value reviews that express opinions different from the consensus, and that will turn out to be correct, more than reviews that are simply confirming the current consensus. 
Informative incentive schemes lead to a quick convergence to the true valuation for all papers. 

\subsection{Informativeness and Accuracy of a Review}

We introduce an incentive system for reviewers that is both truthful and informative. 
Consider the sequence of ratings $x_0, x_1, x_2, \ldots, x_n$ that have been assigned, in chronological order, to a given paper, where $x_0$ is the default rating that is assigned by the system to every paper who is added to the system, as a starting point. 
To define the bonus $\bonus_i$ for the author of rating $x_i$, for $0 < i < n$, let 
\begin{align*}
\qprev_i &= \avg\set{x_0, x_1, \ldots, x_{i-1}} \\
\qcur_i  &= \avg\set{x_{i+1}, x_{i+2}, \ldots, x_n}
\end{align*}
be the averages of the ratings preceding and following $x_i$, respectively.
Let $L$ be the quadratic loss function, defined by $L(a,b) = (a-b)^2$.
We define the {\em accuracy loss\/} and {\em informativeness\/} of the rating $x_i$ as follows:
\begin{align}
\label{eq-informativeness}
& \mbox{Informativeness:}  & \inform_i & = L(\qprev_i, \qcur_i)
\\[1ex]
\label{eq-accuracy} 
& \mbox{Accuracy loss:}  & \truthf_i & = L(x_i, \qcur_i)
\end{align}
To reward reviewers which are {\em both\/} informative and accurate, we let the review bonus $b_i$ be: 
\begin{equation} \label{eq-bonus}
    b_i = \inform_i \cdot \fsigm_{\alpha,M}(\truthf_i)
\end{equation}
where $\fsigm_{\alpha,M}$ is a sigmoidal function parameterized by a parameter $\alpha > 0$, and by the maximum rating $M$ that can be given to a paper (see Figure~\ref{fig:alphas}).
The smaller $\alpha$ is, the stronger the incentive for accuracy is.
The sigmoidal function is such that $\fsigm_{\alpha,M}(0) = 1$, and $\fsigm_{\alpha,M}(M^2) = 0$, so that perfectly accurate reviewers will get their full bonus, and reviewers with the maximum possible value $M^2$ of accuracy loss will not get any bonus.

\begin{figure}[th]
\begin{center}
\begin{tikzpicture}[scale=2.5]
\colorlet{darkred}{red!65!black}
\colorlet{darkblue}{blue!55!black}
\draw (0,0) node[below left] {$0$};
\draw[very thin,color=gray] (-.03,-0.05) grid (3.05,1.05);
\draw (0,0) -- (0,1) node[left] {$1$};
\draw (-0.1,0) -- (3,0) node[below] {$M^2$};
\draw[->] (3,0) -- (3.1,0) node[right] {};
\draw[->] (0,-0.2) -- (0,1.1) node[above] {};
\node at (1.5, -0.1) {$\truthf_i$};
\draw[color=darkred, domain=0.1:0.7] plot (\x,{1/(1+exp(-0.2*(3-\x)/\x+\x/(0.2*(3-\x))))}) 
                                     node[left] {{\small $\fsigm_{0.2,M}(\truthf_i)$}};
\draw[color=darkred, domain=0.1:1.9] plot (\x,{1/(1+exp(-0.2*(3-\x)/\x+\x/(0.2*(3-\x))))});
\draw[color=black, domain=0.4:1.7] plot(\x,{1/(1+exp(-(3-\x)/\x+\x/(3-\x)))}) 
                                   node[left] {{\small $\fsigm_{1,M}(\truthf_i)$}};
\draw[color=black, domain=0.4:2.7] plot(\x,{1/(1+exp(-(3-\x)/\x+\x/(3-\x)))});

%
\draw[color=darkblue, domain=1.1:2.4] plot (\x,{1/(1+exp(-5*(3-\x)/\x+\x/(5*(3-\x))))})
                                     node[right] {{\small $\fsigm_{5,M}(\truthf_i)$}};
\draw[color=darkblue, domain=1.1:2.9] plot (\x,{1/(1+exp(-5*(3-\x)/\x+\x/(5*(3-\x))))});
\draw[color=darkred] (3,0) -- (1.9,0);
\draw[color=black] (3,0) -- (2.7,0);
\draw[color=darkblue] (3,0) -- (2.9,0);
\draw[color=darkblue] (0,1) -- (1.1,1);
\draw[color=black] (0,1) -- (0.4,1);
\draw[color=darkred] (0,1) -- (0.1,1);
\end{tikzpicture}
\end{center}
\caption{Three instances of $\fsigm_{\alpha,M}(x)$, for $\alpha = 0.2$, $\alpha = 1$ and $\alpha = 5$.
\label{fig:alphas}}
\end{figure}
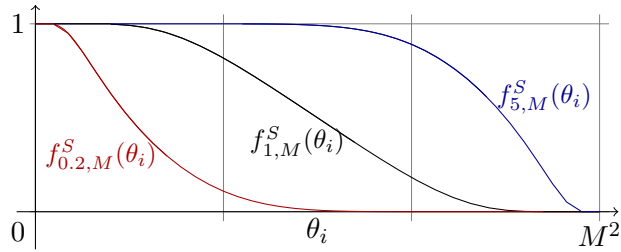

The informativeness, accuracy loss, and bonus defines by (\ref{eq-informativeness})--(\ref{eq-bonus}) have a number of important properties. 

Informativeness provide an incentive to select papers whose current evaluation is most different from what the future consensus will be.
Crucially, the informativeness depends on the {\em previous\/} and {\em future\/} ratings of the paper, but not on the rating $x_i$ assigned by the reviewer under consideration. 
Thus, once the reviewer selects a paper to review, informativeness plays no further role, and the bonus depends entirely on the accuracy loss. 
This decouples choice of paper from accuracy, and will be crucial in proving that the incentive scheme is truthful. 

The accuracy loss is computed by comparing the user's rating only with {\em future\/} valuations.
This eliminates any incentive to provide a valuation that is similar to the known past ones, against the user's true belief about the paper. 
Consider, for instance, an alternative definition where $\qcur$ represents the average of all valuations.
If the user has some prior knowledge about the typical number of reviews a paper is likely to receive, and most of those have already been written, then the user would be able to achieve high accuracy just by providing a rating that is the average of the past ones. 

A consequence of our choice of incentive scheme is that we reward users who discover new information, and present it convincingly in their reviews.
For such users, the score they propose is different from past ones, and influences all future ratings. 
The informativeness will be high, due to the difference between past and future ratings, and the accuracy will not suffer from the difference between the score and previous ratings. 
Notice how this latter property would not hold if we included all ratings in the computation of $\qcur$.

\subsection{Truthfulness} \label{sec:truthful}

Our incentive schemes are not truthful in the strong sense that it is a weakly dominating strategy for players to give ratings that reflect their true opinion of the paper.
There are many collusion schemes that form Nash equilibria where deviating reduces the bonus: for instance, the one where all reviewers provide the same, constant, rating. 
Since there is no ground truth for paper quality independent from reviewer-provided ratings, the inability to ensure that truthful strategies are weakly dominating is unavoidable: reviewers could agree to review a paper as if it were another one; there is nothing intrinsic that ties the reviewer behavior to the paper being reviewed that can be used in the mechanism design.  
The best we can show is that under some conditions, being truthful is a strict Nash equilibrium, that is, a Nash equilibrium from which deviating is not favorable. 

The rating process can be modeled as a Bayesian game \cite{osborne_course_1994}, in which each user $i$ can observe the ratings $x_1, x_2, \ldots, x_{i-1}$ given by previous users to a given paper, as well as their own belief $y_i$ about the paper quality $\qtrue$. 
On the basis of these observations, user $i$ must in turn provide a rating $x_i$ for the paper. 
In formulating our results, we assume that the private estimate $y_i$ of the quality of the paper available to player $i$ is unbiased, and that estimates of different reviewers are uncorrelated. 
The assumption that the user estimates are not overall biased is unavoidable: 
there is no way to distinguish between a paper of quality $\qtrue$, and a paper of quality $\qtrue - \Delta$ for some $\Delta > 0$, 
which all users over-appreciate on average by an amount $\Delta$. 
Put another way, there is no way to differentiate good papers from papers that only seem good to all users: the two notions coincide in our system.

In the following theorems, we use the fact that the bonus received by the $i$-th user is a combination $b_i = F(\truthf_i, \inform_i)$, where $F$ is monotonically decreasing in $\truthf_i$ and monotonically increasing in $\inform_i$. 
As $\inform_i$ is fully determined by the choice of paper, to reason about truthflness, we can reason on the $\truthf_i$ component only.

Our first result concerns reviewers who rate papers without access to other reviews. 
In this case, it is easy to show that being truthful is a strict Nash equilibrium.

\begin{theorem}{} \label{thm:nash1}
Assume all users form statistically uncorrelated and unbiased estimates of the quality of each paper, and assume that users enter their review without being able to read other reviews first.
Then, the strategy profile under which all users rate the paper with their quality estimates is a strict Nash equilibrium.
\end{theorem}

\begin{proof}
Let $0 < i < n$, and assume $x_j = y_j$ for $0 < j \leq n$, $j \neq i$, so that all users except the $i$-th rate papers with their individual estimate, and consider the point of view of the $i$-th user.
The user must minimize $\truthf_i$. 
As the private estimates $\set{y_j}_{0 < j \leq n}$ are uncorrelated, the expected value of $\truthf_i$ can be written as the sum of two variances $v_i + v_f$, where $v_i$ is the variance of $x_i$ with respect to the true value $\qtrue$ of the paper, and $v_f$ is the variance of $\qcur_i$ with respect to $\qtrue$. 
As user $i$ has no influence over $v_f$, the user must minimize $v_i$, and this entails voting the best estimate $y_i$ 
of $\qtrue$ available to the player, so that $x_i = y_i$.
\end{proof}

We can extend this result to the case in which reviewers can read previous reviews, and adjust their submitted ratings according to the previous ratings for the paper. 
Consider again users $1, 2, 3, \ldots$, with private uncorrelated estimates $y_1, y_2, y_3, \ldots$,
whose expected value is the quality $\qtrue$ of the paper.
We assume for simplicity that these private estimates all have the same variance $v$
(the general case is similar, and can be obtained by weighing each estimate with the inverse of its variance). 
In a truthful strategy profile, each user will report the most precise estimates that can be computed from the private information and from the previous ratings. 
Thus, user 1 will report $x_1 = y_1$, user 2 will report $x_2 = (x_1 + y_2)/2 = (y_1 + y_2) / 2$, and in general, user $n$ will report $(n-1)x_{n-1}/n + y_n / n = (y_1 + y_2 + \cdots + y_n) / n$. 
The next theorem shows that deviating from this truthful strategy yields a lower bonus, so that the truthful strategy profile is a Nash equilibrium. 

\begin{theorem}{} \label{thm:nash2}
If reviewers have access to previous reviews, and if their private estimates are uncorrelated, being truthful is a strict Nash equilibrium.
\end{theorem}

\begin{proof}
Consider users $1, 2, \ldots, n, n+1$, with uncorrelated private estimates $y_i$. 
We show that it is optimal for user $n$ to be truthful; the general case for $1 \leq i \leq n$ is similar but leads to more complicated notation. 
If user $n+1$ plays truthfully, while player $n$ deviates from the truthful amount by $\Delta$, we have:
\[
  x_n = \frac{n-1}{n} x_{n-1} + \frac{y_n}{n} + \Delta
  \qquad
  x_{n+1} = \frac{n}{n+1} x_n + \frac{y_{n+1}}{n+1} \eqpun .
\]
The expected loss $\E{(x_n - x_{n+1})^2}$ is thus equal to
\begin{equation} \label{eq:exploss}
    \E{\Bigl(
    \frac{n-1}{n(n+1)} x_{n-1} + \frac{y_n}{n(n+1)} + \frac{\Delta}{n+1} - \frac{y_{n+1}}{n+1} 
    \Bigr)^2} \eqpun .
\end{equation}
Noting that $\E{x_{n-1}} = \E{y_n} = \E{y_{n+1}} = \qtrue$, 
we have that the coefficient of $\Delta$ in the expansion of (\ref{eq:exploss}) is
\[
    \frac{\qtrue}{(n+1)^2} \Bigl[ \frac{n-1}{n} + \frac{1}{n} - 1 \Bigr] = 0 \eqpun .
\]
Thus, (\ref{eq:exploss}) depends on $\Delta$ only via $\Delta^2 / (n+1)^2$.
Since $x_{n-1}$, $y_n$, and $y_{n-1}$ are mutually uncorrelated, we obtain from (\ref{eq:exploss}):
\[
  \E{(x_n - x_{n+1})^2} = \frac{v}{n(n+1)} + \frac{\Delta^2}{(n+1)^2} \eqpun ,
\]
where $v$ is the variance of one of the individual estimates $y_i$. 
Thus, user $n$ incurs minimum loss when $\Delta = 0$, showing that deviating from truthful behavior reduces the review bonus. 
Intuitively, any variation from the truth by one user only partially influences later users, raising the loss of the deviating user.
\end{proof}

\section{Simulations} \label{sec:sim}

We have shown in the previous section that the accuracy part of the incentive ensures that,
once a reviewer has chosen a paper to evaluate, it is in her best interest
to evaluate it honestly.
It remains to show that the informativeness term of the incentive 
encourages users to choose papers in a way that benefits the overall quality of the ranking.
We provide evidence in this direction through a set of simulations
in which a population of 1000 users evaluates a collection of 1000 papers.

We assume that each paper has an intrinsic quality $\qtrue$ which represents
our ground truth.
At any given time, the system attributes a current rating to each paper.
Such rating starts at zero and is updated as the arithmetic average
of the grades provided by the reviewers (including the initial default value of zero).

The reputation resulting from a review is defined by the bonus~\eqref{eq-bonus}.
The core component of the simulation is its \emph{user model},
dictating how simulated users choose a paper to review and a grade for it.
In particular, simulated users hold certain \emph{beliefs}
about the papers, which allow them to estimate the expected reputation boost
deriving from reviewing a certain paper.
Supported by the observations in the previous sections, 
we assume that users grade papers
truthfully, i.e., according to the best reconstruction allowed by the model.

\subsection{User Models}

We stipulate that each user is interested in a random sample of 100 papers out of the total 1000.
On each of those papers, the user initially holds the following beliefs:
the paper quality $z$ and the corresponding expected error $\sigma$.
One can think of $\frac{1}{\sigma}$ as the \emph{competence} of the user
on that paper, of which the user is self-aware.
Moreover, users are aware of the average error $\bar{\sigma}$ among all users and all papers.

Next, we describe how the above parameters are sampled.
The true value $\qtrue$ is sampled for each paper out of a normal distribution.
Then, each user is attributed a typical error $\sigma^t$ 
out of a distribution with mean $\bar{\sigma}$; the typical error 
indicates the ``overall competence'' of the user for the papers under consideration.
The paper-specific error $\sigma$ is sampled
from a distribution with mean $\sigma^t$.
Finally, the perceived paper quality $z$ is sampled out of a normal distribution 
with mean $\qtrue$ and standard deviation $\sigma$ (denoted by $\norm(\qtrue, \sigma)$).
Thus, we model users of varying degrees of average competence,
and with each a set of papers that they might consider reviewing. 

We present results for two user models. 
The models coincide on the original beliefs held by the users about the papers, but differ in the way users take into account previous reviews received by a paper.
In the first user model, users grade according to their belief, 
without taking into account previous reviews.
In the second user model, users revise their belief 
to take into account the grades in the previous reviews and their supposed accuracy.
In both models, each user starts with the belief $(z, \sigma)$ described above,
for each paper in which she is interested.

\paragraph*{First user model.}
In the first user model, reviewers believe their quality estimate $z$
to be the best estimate for the future consensus grade assigned by the system to that paper.
Accordingly, accuracy is simply estimated as $\sigma^2$ and informativeness as $(\qprev - z)^2$,
where $\qprev$ is the average of the previous grades,
leading to the reputation boost estimate
$$(\qprev - z)^2 \cdot \fsigm_{\alpha,M}(\sigma^2).$$
The above estimate is used to choose which paper to review.
Once a given paper is chosen to be reviewed, it will receive grade $z$.
This user model is consistent with the assumptions of Theorem~\ref{thm:nash1}.

\paragraph*{Second user model.}
In our second user model, users look at previous reviews 
to reconstruct via Bayesian inference the most likely grade for a paper. 
Based on these beliefs and taking into account previous reviews,
users estimate the reputation boost they may receive from evaluating a given paper.
Consider a paper with $n$ previous reviews and current evaluation $\qprev$.
Since a user does not hold a specific belief on the competence of previous reviewers, she assumes they all share the same error $\bar\sigma$.
Then, in this user model,  the best quality estimate $\hat{q}$ and the corresponding error $\hat{\sigma}$ are obtained by Bayesian inference with prior $\norm(z, \sigma)$ and observation $\qprev$ with likelihood $\norm(\qtrue,\frac{\bar{\sigma}}{\sqrt{n}})$.
The likelihood follows from assuming that previous reviewers adopted a similar Bayesian inference procedure, starting from statistically independent private beliefs.
Accordingly, the accuracy term of the incentive is estimated as $\hat{\sigma}^2$ and informativeness as $(\qprev - \hat{q})^2$, leading to the reputation boost estimate 
$$
(\qprev - \hat{q})^2 \cdot \fsigm_{\alpha,M}(\hat{\sigma}^2).
$$
Once the user chooses a paper, he will enter the grade $\hat{q}$ for it.  
This user model is consistent with the assumptions of Theorem~\ref{thm:nash2}, except that here each user is aware of its own variance ($\sigma$)
and assumes that all previous reviewers have the same variance $\bar\sigma$.
The same-variance assumption for previous reviewers is motivated by the fact that
we envision reviews to appear anonymous, so that users cannot infer the variance of a review
from the identity of its author.

\subsection{Choice of paper to review} \label{sec:choices}

At each round, a user is selected in round-robin fashion and performs a truthful review of a paper.
%
We compare three different scenarios in which users choose which paper to review in the following ways:
\begin{itemize}
\item {\bf Random:} uniformly at random among the papers known to the user.
\item {\bf Selfish:} the user chooses the paper that maximizes the estimated reputation boost,
as described in the user model, with a varying value of $\alpha$.
\item {\bf Accuracy:} the user chooses the paper that minimizes the estimated accuracy loss.
\item {\bf Informativeness:} the user chooses the paper that maximizes the estimated informativeness loss.
\item {\bf Optimal:} the user chooses the paper that maximizes global loss decrease,
assuming that she knows the real quality of all papers, but still grades according to her beliefs.
\end{itemize}
The ``random'' criterion is used to measure the performance of our incentive system compared to a system where users, lacking incentives, pick the paper they wish to review uniformly at random. 
The ``accuracy'' and ``informativeness'' criteria are used to show that an appropriate combination
of these two (a.k.a.\ selfish choice) is more effective than either of them separately.
The ``optimal'' criterion is deliberately irrealistic
and meant to serve as a reference for the fastest possible global loss decrease compatible with user beliefs about paper quality.

\subsection{Performance Measures}

Our first performance measure is the \emph{global loss} of the current quality estimates, computed
as the sum over all papers of the squared difference between the current 
paper quality estimate and the paper intrinsic quality $\qtrue$.

To illustrate how more expert reviewers receive more reputation (total review bonus points),
we also report in Figures~\ref{tab:corr1} and~\ref{tab:corr2}
the Pearson and Spearman correlations between user competence 
(the user-typical error $\sigma^t$ discussed above) 
and the reputation at the end of the experiment. 
Additionally, the fourth column in Figures~\ref{tab:corr1} and~\ref{tab:corr2}
reports the expected error incurred by a user when grading a paper, relative to the typical
error of that user. A value close to $1$, such as the one obtained by the random choice criterion,
implies that users select papers independently of their specific competences.
On the contrary, the lower the value the more users are choosing the papers they are more familiar with.
Since users in practice are likely to prefer those papers anyway, we see a low
value in that column as a desideratum for our incentive scheme.
The relative error values are averaged over rounds and data sets, 
and the last column in our tables displays the standard deviation over data sets.

\begin{figure*}[th]
\centering \includegraphics[scale=0.9]{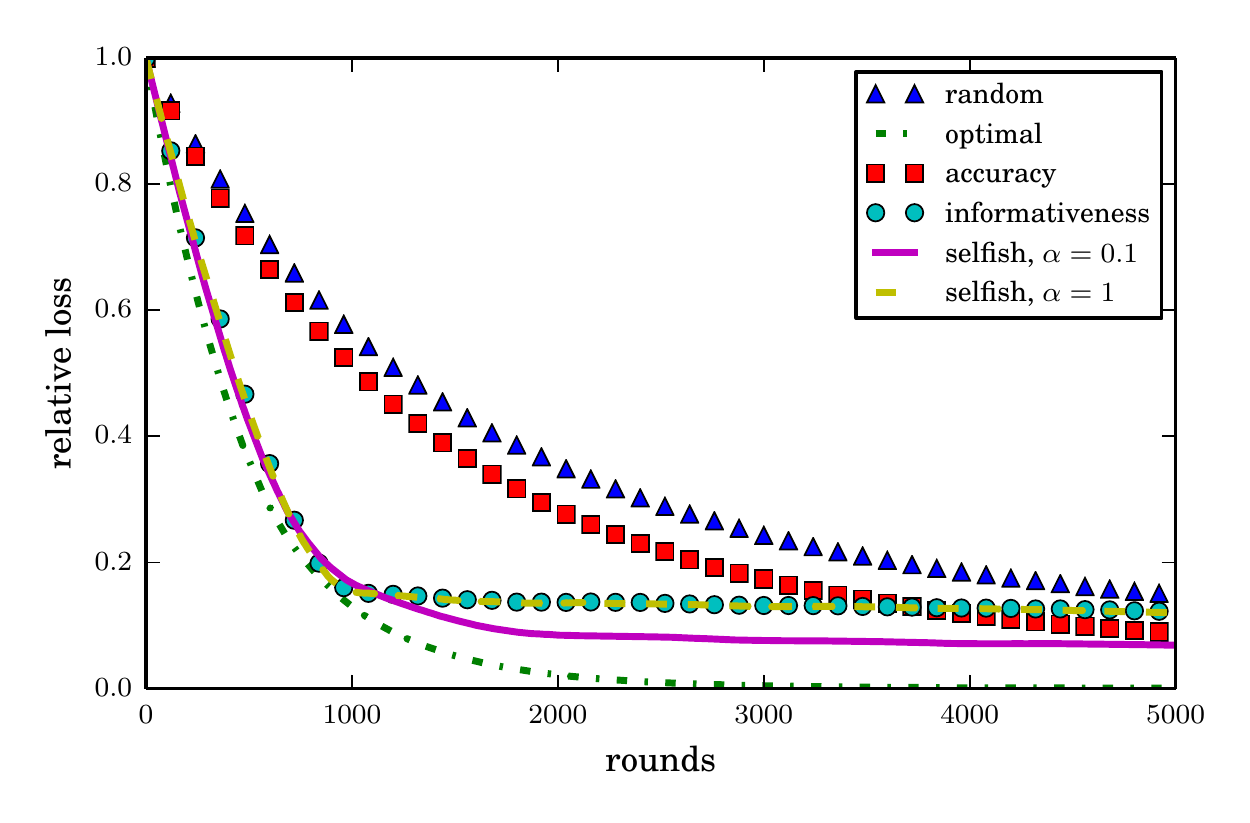}
\caption{Relative global loss in the first user model.}
\label{fig:loss1}
\end{figure*}
    
\begin{figure*}
\centering \includegraphics[scale=0.9]{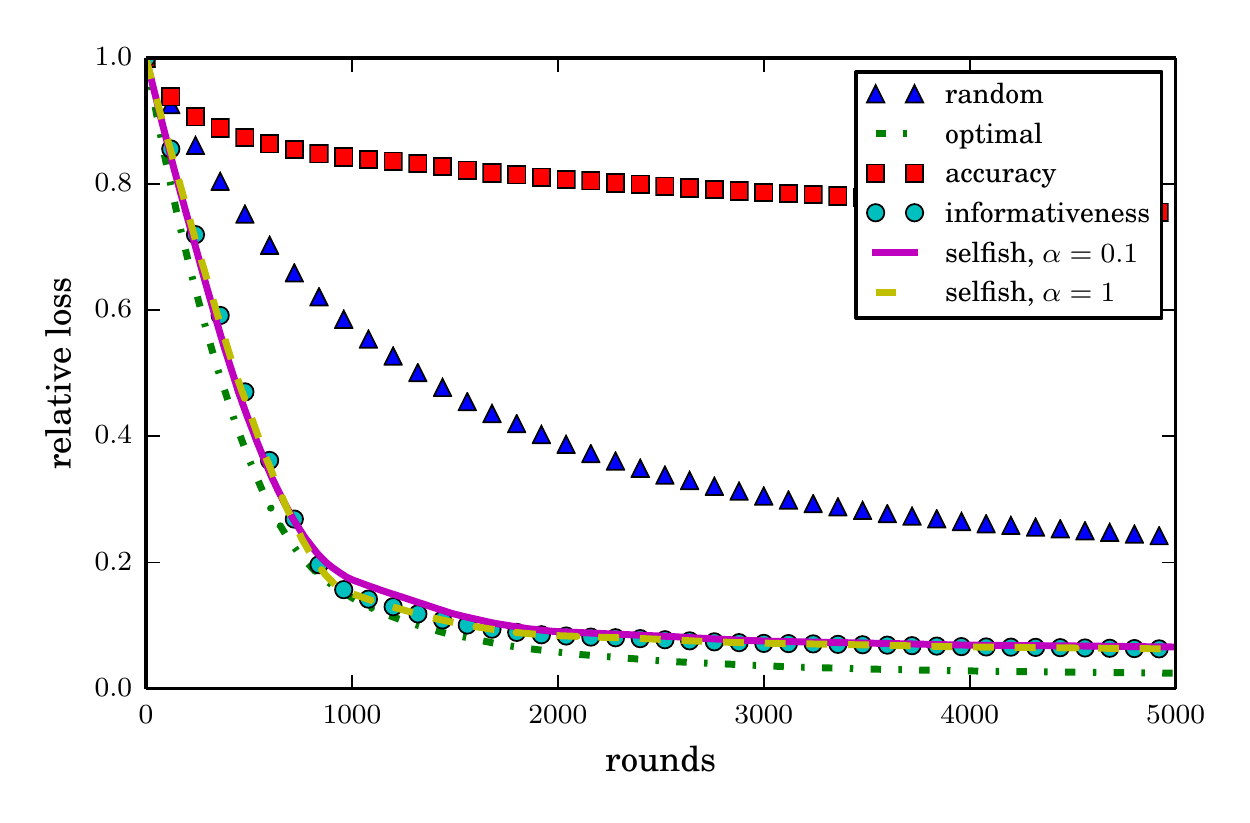}
\caption{Relative global loss in the second user model.}
\label{fig:loss2}
\end{figure*}

\begin{figure*}
\centering \includegraphics[scale=0.75]{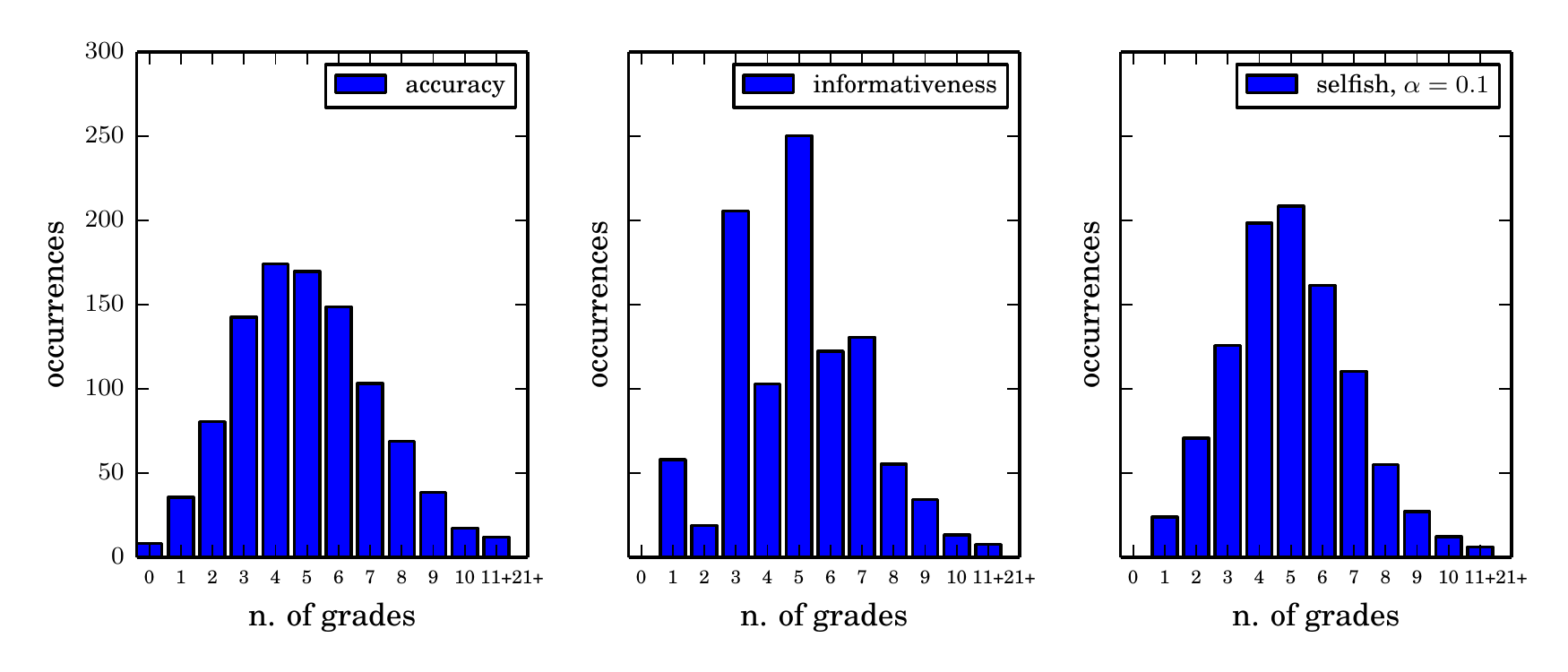}
\caption{Distribution of the number of grades per paper in the first user model.
The labels 11+ and 21+ stand for the intervals $[11,20]$, $[21, \infty)$, respectively.}
\label{fig:n-grades1}
\end{figure*}

\begin{figure*}
\centering \includegraphics[scale=0.75]{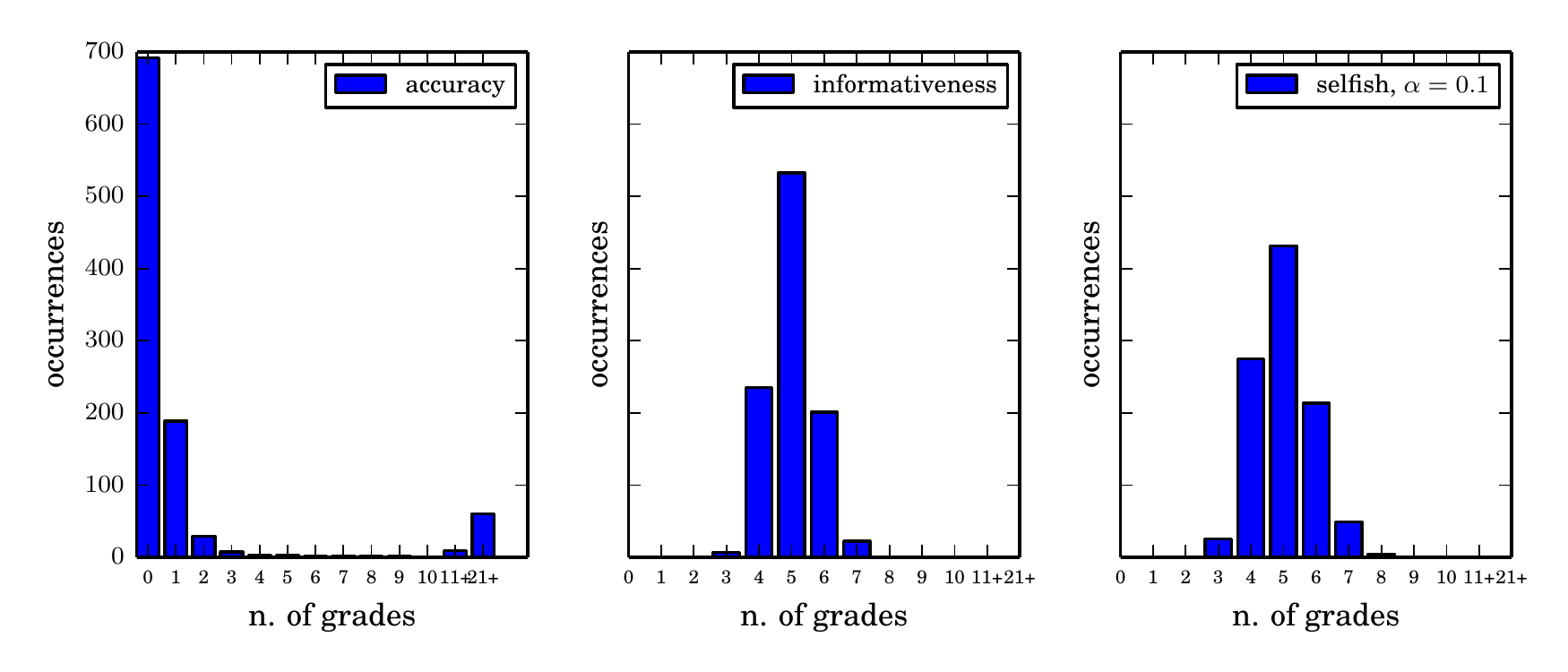}
\caption{Distribution of the number of grades per paper in the second user model.
The labels 11+ and 21+ stand for the intervals $[11,20]$, $[21, \infty)$, respectively.}
\label{fig:n-grades2}
\end{figure*}

\begin{figure*}
\begin{centering}
\begin{tabular}{l||r||r|r||r|r}
choice criterion       &loss  &Pearson &Spearman &rel.\ error &rel.\ error std.\ dev.\\
\hline
random                 &0.24  &-0.005   &-0.006  &0.999   &0.003 \\
accuracy               &0.17  &0.011   &0.011  &0.536   &0.003 \\
informativeness        &0.13  &0.069   &0.063  &1.049   &0.003 \\
optimal                &0.00  &0.015   &0.014  &0.972   &0.003 \\
selfish, $\alpha=0.1$  &0.08  &0.162   &0.161  &0.719   &0.002 \\
selfish, $\alpha=1$    &0.13  &0.070   &0.070  &1.030   &0.004 \\
\end{tabular}
\caption{Summary data for the first user model.
The columns contain: the relative global loss after 3000 reviews;
the Pearson and Spearman correlations between competence and reputation after 5000 reviews;
the average relative error (defined in Section ``Performance Measures'') and its standard deviation across 10 data sets.}
\label{tab:corr1}
\end{centering}
\end{figure*}

\begin{figure*}
\begin{centering}
\begin{tabular}{l||r||r|r||r|r}
choice criterion       &loss   &Pearson     &Spearman &rel.\ error &rel.\ error std.\ dev.\\
\hline
random                 &0.30   &-0.002      &-0.004     &1.002  &0.004\\
accuracy               &0.79   &0.323       &0.228      &0.944  &0.003\\
informativeness        &0.07   &0.020       &0.020      &0.769  &0.002\\
optimal                &0.04   &-0.028      &-0.030     &0.806  &0.004\\
selfish, $\alpha=0.1$  &0.08   &0.157       &0.155      &0.740  &0.004\\
selfish, $\alpha=1$    &0.07   &0.070       &0.067      &0.782  &0.004\\
\end{tabular}
\caption{Summary data for the second user model. 
  The columns have the same interpretation as Figure~\ref{tab:corr1}.}
\label{tab:corr2}
\end{centering}
\end{figure*}

\subsection{Results}

We simulated the behavior of 1000 users evaluating a set of 1000 papers.
Each user holds beliefs on a random subset of 100 papers that the user is willing to review.
We simulate 5000 reviewing rounds. 
We repeated each simulation run 10 times, in order to measure the standard deviation of the results across the runs.

\subsubsection{Results for First User Model}

Figure~\ref{fig:loss1} shows the value of the global loss, relative to the initial global loss,
when paper choice is performed according to the five criteria from Section~\ref{sec:choices}.
The first column of the table in Figure~\ref{tab:corr1} reports the loss after 3000 reviews.
The other columns contain
the Pearson and Spearman correlations between user competence and reputation,
and the user propensity for reviewing papers on which they are most proficient. 

The relative global loss curve shows that the selfish choice, for $\alpha \in \{0.1, 1\}$,
performs very well and close to the optimal choice of papers, especially in the first
1000 rounds of the experiments.
In fact, a closer inspection reveals that, when papers have the default starting score of 0 
and no reviews yet, users simply choose papers with high perceived quality and no reviews,
in order to reap a large informativeness bonus.
Hence, at the beginning many papers go from the default score of 0 to approximately $\frac{M}{2}$,
justifying the initial steep decline in global loss.

Then, consider the curve for the ``informativeness'' choice criterion.
After the first 1000 rounds, when the above phenomenon leads to a near-optimal performance,
the curve is essentially flat. Indeed, when users are only incentivized to provide
informative grades, they will preferably select papers for which they have a very extreme opinion
(very low or very high), leading to oscillation of paper scores, rather than convergence
to the true value.

Notice that the global loss curve for $\alpha=1$ is very close to the one for
the ``informativeness'' choice criterion. 
This is due to the fact that the sigmoid $\fsigm_{1,M}(x)$ stays very close to 1
for relatively large values of $x$.
Roughly speaking, the sigmoid ``forgives'' large accuracy errors.
Hence, even if a user expects a significant accuracy loss, she can count on a reward
almost equal to the expected informativeness bonus.
Specifically, in our experiments we have $M=10$ (grades between 0 and 10) and users have 
a maximum standard deviation of $5$. 
So, their estimate for the accuracy loss is in the range $[0, 25]$, which
corresponds to very limited reward penalties ($\fsigm_{1,10}(25) \approxeq 0.94$).

On the other hand, when $\alpha=0.1$ even a small accuracy loss incurs a significant
penalty on the reward, so the two components of the incentive are properly balanced.
The global loss curve is initially steep and competitively positioned w.r.t.\ both 
the ``optimal'' curve and the curve based on accuracy alone.
This is confirmed by Figure~\ref{tab:corr1}, reporting the relative error 0.719 for this case and
a moderate correlation of 0.161 between competence and final reputation,
higher than all other cases.
Summarizing, data from this user model suggests that a choice of 
$\alpha$ close to 0.1 might be appropriate to the parameters of our populations.



\subsubsection{Results for Second User Model}

Figures~\ref{fig:loss2} and~\ref{tab:corr2} show the relative global loss and the other performance
measures for this model.

It may appear surprising that the choice based on accuracy alone
performs even worse than the random choice.
Indeed, when accuracy is the only incentive, users tend to focus on papers that have already received many reviews, because their quality can be more accurately predicted on the basis of the previous ratings. 
This creates a perverse incentive, in which the papers whose quality is best known draw the most evaluations.
Figure~\ref{fig:n-grades2} confirms that in that case more than 50 papers receive a very large amount of ratings, whereas 700 papers are completely neglected.
The distribution of the number
of grades per paper becomes much more balanced with the selfish choice and $\alpha = 0.1$, 
when the informativeness term mitigates the above issue.

Similarly to the other user model, our incentives with $\alpha = 0.1$ display the best overall
performance, with 8\% global loss after 3000 reviews, positive correlation between competence
and reputation, and low relative error of 0.74, 
proving a clear bias for choosing papers on which the user is particularly competent.

In practice, this set of experiments suggests that the proposed incentive scheme may provide
strong advantages, compared to rewarding accuracy alone, once it has been properly tuned
to the characteristics of the user and paper populations.

Comparing Figures \ref{fig:loss1} and \ref{fig:loss2}, we note that even in the optimal case, the global loss decreases faster for the first user model than for the second one. 
This can be explained by noting that in the first user model, users grade papers according to their individually-formed opinion, without access to other user's reviews.
If $n$ users provide grades for a paper, and the grades are then averaged, the individual opinions of each user account for $1/n$ of the average, which is optimal lacking information on the accuracy of individual users.
In the second user model, instead, users use Bayesian inference to improve the accuracy of their estimate on the basis of reviews of previous users. 
As a consequence, the individual estimate $\hat{q}_i$ of the $i$-th user accounts only 
for $1/i$ of the grade provided by user $i$ (assuming constant variances), 
and for $\frac{1}{n} \cdot \sum_{k=i}^n \frac{1}{k}$ of the complete average grade. 
This non-uniform weighing of the individually formed opinions is not optimal, and slows loss decrease.

\section{Discussion}

In this white paper, we advocate a shift from pre to post-publication peer review for scientific papers.
The chief benefit of post-publication peer review is the more timely circulation of scientific ideas, which can be shared as soon as the authors decide to publish them. 
The key to a successful process of post-publication peer review consists in creating venues where authors are willing to post their papers for review, and where reviewers are incentivized to do useful and fair review work. 

To facilitate this, we are proposing to create a tool, TrueReview, in support of post-publication peer review. 
TrueReview will allow people to set up new venues where papers can be submitted (for example, corresponding to conferences or special topics), as well as venues that index papers appearing in open-access venues such as arXiv. 
To encourage useful and accurate reviews, TrueReview will list with similar prominence both papers and reviewers: the papers will be ranked according to their quality, as assessed by the reviewers, and the reviewers will be ranked in order of the total {\em review bonus\/} they have accrued. 
The review bonus thus works as an incentive for reviewers.

We propose to award review bonus according to a combination of review {\em accuracy\/} and {\em informativeness.}

The accuracy measures the precision of a review's evaluation, in light of future evaluations. 
Judging a review only in view of future ones is instrumental in creating a truthful incentive for reviewers, where expressing their own best judgement on the paper's quality is an optimal strategy. 
Furthermore, measuring the accuracy of a review by comparing it with future reviews only rewards people who discover significant facts about papers, explain them in their review, and thereby influence future reviews. 

The informativeness of a review is a measure of how much the future evaluation of a paper differs from the current one.
Awarding a bonus for informativeness thus creates an incentive for reviewers to select papers who have received no or few reviews, or whose reviews are grossly imprecise. 
As the informativeness of a review is unrelated to the rating expressed in the review itself, including informativeness in the bonus does not alter the truthful nature of the incentive schemes. 

We combine the accuracy and informativeness schemes in a multiplicative fashion, such that reviewers need to be {\em both\/} accurate and informative in order to obtain a bonus. 
This prevents lazy review strategies, such as picking papers with a large number of reviews and simply restating the consensus opinion on these papers (accurate but not informative), or picking new papers and just entering a random review (informative but not accurate).

We have experimented with two users models: one in which users base their review on their opinion only, and another in which they examine and account for previous ratings, before forming their opinion of the paper's quality. 
For both user models, our experiments show that the review bonus that combines both informativeness and accuracy is superior to considering either informativeness or accuracy alone, and is superior also to offering no specific bonus, and resorting on simpler methods such as simply counting how many reviews each user has provided.

\bibliography{paper-review}
\bibliographystyle{alpha}

\end{document}